\documentclass[journal,10pt,twocolumn,twoside]{IEEEtran}

\usepackage{cite}

\ifCLASSINFOpdf
   \usepackage[pdftex]{graphicx}
\else

\fi

\usepackage[cmex10]{amsmath}
\usepackage{breqn}
\usepackage{algorithm}
\usepackage{algcompatible}
\usepackage{array}
\usepackage[caption=false,font=normalsize,labelfont=sf,textfont=sf]{subfig}

\usepackage{url}

\usepackage{xcolor}

\newcommand\NoThen{\renewcommand\algorithmicthen{}}

\definecolor{Darkgray}{gray}{0}

\usepackage{multicol,multirow}
\usepackage{cleveref}

\usepackage{amsthm}
\usepackage{makecell}
\usepackage{mathtools}

\newtheorem{proposition}{\bf Proposition}

\newcommand{\ve}[1]{\boldsymbol{#1}}

\newcolumntype{I}{!{\vrule width 1.2pt}}
\makeatletter
\def\hlinewd#1{%
\noalign{\ifnum0=`}\fi\hrule \@height #1 %
\futurelet\reserved@a\@xhline}

\begin{document}
\title{An Incentive-compatible Energy Trading Framework for Neighborhood Area Networks with Shared Energy Storage}

\author{Chathurika~P.~Mediwaththe,~\IEEEmembership{Member,~IEEE,} Marnie Shaw, Saman Halgamuge,~\IEEEmembership{Fellow,~IEEE}, David B. Smith,~\IEEEmembership{Member,~IEEE,} and Paul Scott%
\thanks{C. P. Mediwaththe, M. Shaw, S. Halgamuge, and P.~Scott are with the Australian National University, Canberra, ACT 0200, Australia.
e-mail: (chathurika.mediwaththe, marnie.shaw, saman.halgamuge, paul.scott@anu.edu.au).}
\thanks{S. Halgamuge is also with the University of Melbourne, Parkville, VIC 3010, Australia.}
\thanks{D. B. Smith is with Data61, CSIRO, Eveleigh, NSW 2015, Australia. e-mail: (david.smith@data61.csiro.au). }}

\markboth{ IEEE Transactions on Sustainable Energy}%
{ IEEE Transactions on Sustainable Energy}

\maketitle
\begin{abstract}
Here, a novel energy trading system is proposed for demand-side management of a neighborhood area network (NAN) consisting of a shared energy storage (SES) provider, users with non-dispatchable energy generation, and an electricity retailer. In a leader-follower Stackelberg game, the SES provider first maximizes their revenue by setting a price signal and trading energy with the grid. Then, by following the SES provider's actions, the retailer minimizes social cost for the users, i.e., the sum of the total users' cost when they interact with the SES and the total cost for supplying grid energy to the users. A pricing strategy, which incorporates mechanism design, is proposed to make the system incentive-compatible by rewarding users who disclose true energy usage information. A unique Stackelberg equilibrium is achieved where the SES provider's revenue is maximized and the user-level social cost is minimized, which also rewards the retailer. A case study with realistic energy demand and generation data demonstrates 28\%~-~45\% peak demand reduction of the NAN, depending on the number of participating users, compared to a system without SES. Simulation results confirm that the retailer can also benefit financially, in addition to the SES provider and the users.
\end{abstract}


\IEEEpeerreviewmaketitle

\section{Introduction}
Electricity demand-side management with distributed energy resources helps accommodate peak electricity demand without the need for upgrading conventional power grid infrastructure. In particular, shared energy storage (SES) systems such as community energy storages \cite{CES_owner} can be effectively used to exploit user-owned non-dispatchable energy generation, for example, rooftop solar or wind energy generation, to regulate users' peak electricity demand. With this capability, SESs would enable electricity retailers to effectively serve the energy needs of neighborhood area networks (NANs) that consist of small groups of residential buildings, for example, 50 households. In addition, SESs can facilitate user-centric demand-side management solutions without end-users having to invest in personal energy storage devices. Increasing interest in SES systems has led to a range of projects worldwide exploring their feasibility for utilizing user-owned non-dispatchable energy generation for demand-side management \cite{Anatolia}.

User-centric demand-side management driven by electricity retailers requires the knowledge of true energy usage information of users, including their energy demand and generation, for robust operation. However, users might intentionally misreport private information so as to gain personal cost benefits that would challenge achieving system-wide objectives in demand-side management \cite{cheating_cus}. Hence, smart pricing strategies capable of motivating users to reveal true private information are imperative for effective demand-side management \cite{samadi1, Ma}.

In this paper, a novel energy trading system is proposed for demand-side management of a NAN that consists of an SES provider, users with non-dispatchable energy generation, and an electricity retailer. The SES is used to store surplus energy from users' non-dispatchable generation and discharged when electricity demand is high. The energy transactions between the SES and the users as well as the energy transactions between the power grid and both the users and the SES are coordinated by the retailer. The interplay between the SES provider and the retailer is formulated as a non-cooperative Stackelberg game where the SES provider (leader) moves first to maximize revenue by setting a price signal and trading energy with the grid through the retailer. Then, by following the SES provider's actions, the retailer (follower) minimizes social cost for the users and determines energy amounts for the users to trade with the grid and the SES. Here, the social cost is defined as the sum of the total users' cost  when they interact with the SES and the cost incurred by the retailer in supplying grid energy to the users. Insights from the Vickrey-Clarke-Groves (VCG) mechanism are applied to the retailer's social cost minimization problem to develop a pricing strategy that can motivate users to reveal true energy usage information to the retailer. This paper has the following contributions:
\begin{itemize}
\item The Stackelberg game has a unique pure strategy equilibrium where the SES provider maximizes revenue and the retailer minimizes the social cost of participating users.
\item The system is incentive-compatible as participating users can only minimize their personal energy costs by revealing their true energy usage information.
\item  In addition to providing financial benefits to the SES provider and the participating users, our solution also enables the retailer to benefit financially by coordinating the energy transactions between the grid and the system.
\end{itemize}
Mechanism design has been widely applied to explore demand-side management that achieves socially desirable outcomes while ensuring incentive-compatibility for users. For instance, mechanism design has been utilized in \cite{Nekouei} to develop a pricing strategy that rewards users who reveal true private information in a load curtailment scheme that minimizes aggregate user inconvenience due to load curtailment. In \cite{samadi1}, mechanism design has been used to develop a demand-scheduling program that minimizes the aggregate user cost and provides benefits to users who reveal true energy usage information. To the best of our knowledge, this paper is the first to leverage mechanism design with Stackelberg game theory to study the feasibility of integrating an SES system with user-owned non-dispatchable energy generation for demand-side management of a NAN by minimizing social energy cost and assuring incentive-compatibility for the users. 

This work has three key differences to our previous work \cite{chathurika2}. First, in contrast to allowing users to minimize personal costs selfishly, the proposed system minimizes total users' energy costs from a social planner's perspective. Second, the users in the proposed system interact with the SES provider through the retailer as a result of minimizing the social cost whereas users in the system in \cite{chathurika2} directly interact with the SES provider. In addition, the solution here is capable of benefiting the retailer for coordinating grid energy with the users and the SES unlike the solution presented in \cite{chathurika2}.

The rest of the paper is organized as follows. Section \ref{sec:2} presents related work, and Section \ref{sec:3} describes system models of the energy trading system. Section \ref{sec:4} explains the formulation of the energy trading system, and Section \ref{sec:5} presents simulation results. Section \ref{sec:6} concludes the paper.
\section{Related Work}\label{sec:2}
The concept of sharing energy storage systems for demand-side management has gained attention in both industry and research communities \cite{Bale,Paridari,Tushar_Joint, EnergyBuildings,Anatolia,TusharDavid}. For instance, sharing household-distributed energy storages through joint-ownership between domestic users and network operators has been explored in \cite{Bale} to facilitate demand response. Sharing user-owned energy storages for demand response has been studied in \cite{Tushar_Joint} while demonstrating an optimal policy to determine the shareable energy storage capacity to network operators. Various algorithms based on mixed-integer linear programming have been proposed for optimal scheduling of electric appliances of users using SES systems \cite{Paridari, EnergyBuildings}. 

Mechanism design and game theory have been widely used to investigate energy management problems, including demand-side management, with strategic user behavior \cite{NLiu,Mhanna,Low,Ma,MohensianGT,Atzeni,HillDong}. Utilizing auction mechanisms with block-chain technology has gained growing interest in peer-to-peer energy trading markets, including energy trading among electric vehicles, to elicit true local information from users \cite{blockchain}. An economically efficient pricing strategy for wholesale electricity markets is proposed in \cite{Low} using the VCG mechanism. Mechanism design has been used to develop a pricing strategy in \cite{Ma} that effectively motivates users to report true energy information to the electricity retailer in an energy consumption scheduling game among users. In contrast, this paper applies insights of mechanism design in a different demand-side management setting where users with non-dispatchable energy generation trade energy with an SES. In the literature, combining incentive compatibility for users in social cost minimization generally results in bi-level structures where users, at the second level, react to the decisions taken by the social planner \cite{Ma,Low,Mhanna,samadi1}. However, in this paper, combining incentive compatibility for users and a Stackelberg game between the retailer and the SES provider leads to a tri-level structure where users, at the third level, implicitly react to the decisions taken in the Stackelberg game.

\begin{table*}[t]
\centering
\caption{table of notation.}
  \begin{tabular}{|m{2 cm}|m{6 cm}||m{2 cm}|m{6 cm}|}
 \hline
Variable/parameter & Definition & Variable/parameter & Definition  \\
  \hline\hline
  
  $\mathcal{A}$ & Set of all users in the NAN. & $\mathcal{P}^{+}(t)$   & Set of surplus energy users at time $t$.  \\
  $\mathcal{P}$ & Set of participating users. & $\mathcal{P}^{-}(t)$ & Set of deficit energy users at time $t$. \\
  
   $\mathcal{N}$ & Set of non-participating users & $e_n(t)$ & Energy traded by user $n$ with the grid at time $t$.  \\
    
  $\mathcal{T}$ &  Entire time period of analysis. & $s_n(t)$ & Surplus energy of user $n$ at time $t$.\\
    
    $H$ & Total number of time steps in $\mathcal{T}$. & $e_s(t)$ & Energy exchanged between the grid and the SES at time $t$. \\
    
   $I$ & Total number of users in $\mathcal{P}$.  & $b(t)$  & The charge level of the SES at the end of time $t$.\\
    
    $\eta^+,~\eta^-$ & Charging and discharging inefficiencies of the SES. & $E(t)$ & Total energy load on the grid at time $t$.\\
    
   $\alpha$ & Leakage rate of the SES; $(0<\alpha \leq1)$. & $E_{\mathcal{A}}(t)$  & Total grid load of the users $\mathcal{A}$ at time $t$.\\ 
    $Q_M$ & Maximum energy capacity of the SES. & $E_{\mathcal{P}}(t)$  & Total grid load of the users $\mathcal{P}$ at time $t$. \\ 
     $E_{\text{max}}$ & Maximum load that the grid can support without overloading it. & $E_{\mathcal{N}}(t)$  & Total grid load of the users $\mathcal{N}$ at time $t$.  \\
     
  $R$ & Revenue of the SES provider. & $\mathcal{E}_n(t)$   & Feasible strategy set of user $n$ at time $t$. \\
    
    $\mathcal{Q}$ & Feasible strategy set of the SES provider. & $k_n(t)$ & Payment made by user $n$ to the retailer at time $t$.\\
    
    $\mathcal{G}$ & Stackelberg game between the SES provider and the retailer. & $C_n(t)$ &  Personal energy cost of user $n$ at time $t$.  \\
    
     $\tau$ & Small positive value. & $\ve{\hat{s}}(t)$    & Declared surplus energy profile by the users $\mathcal{P}$ to the retailer at time $t$. \\
    
   $r$ & Iteration number. & $\ve{\hat{s}}_{-n}(t)$ & Declared surplus energy profile by all other users $\mathcal{P}\backslash n$ to the retailer at time $t$.  \\
  
      $\mathcal{L}$ & SES provider. & $U(t)$  &  Net payoff of the retailer at time $t$.\\
 $\mathcal{F}$ & Retailer. & $ p_g(t)$  & Unit energy price paid by the retailer to the grid at time $t$.  \\
    $d_n(t)$ & Electricity demand of user $n$ at time $t$. & $p_s(t)$    &  Unit energy price of the SES at time $t$. \\
   $g_n(t)$ & Energy generation of user $n$ at time $t$.  & $\delta_t,~\phi_t$ & Positive time-of-use tariff constants at time $t$. \\
  $x_n(t)$ & Energy traded with the SES by user $n$ at time $t$.  & $C_t$  & Social cost at time $t$. \\
 \hline
  \end{tabular}\label{notation}
\end{table*} 
\section{System Models}\label{sec:3}
The energy trading system consists of energy users, an electricity retailer, and an SES owned by a third-party that provides storage services \cite{CES_owner}, referred to as the SES provider, as depicted in Fig. \ref{fig:systemModel}. This section explains the role of each entity and the energy cost models used in the system. The definitions of notations in this paper are summarized in Table~\ref{notation}. \vspace{-0.7cm}
\begin{figure}[b!]
\centering
\vspace{-0.5cm}
\includegraphics[width=0.8\columnwidth]{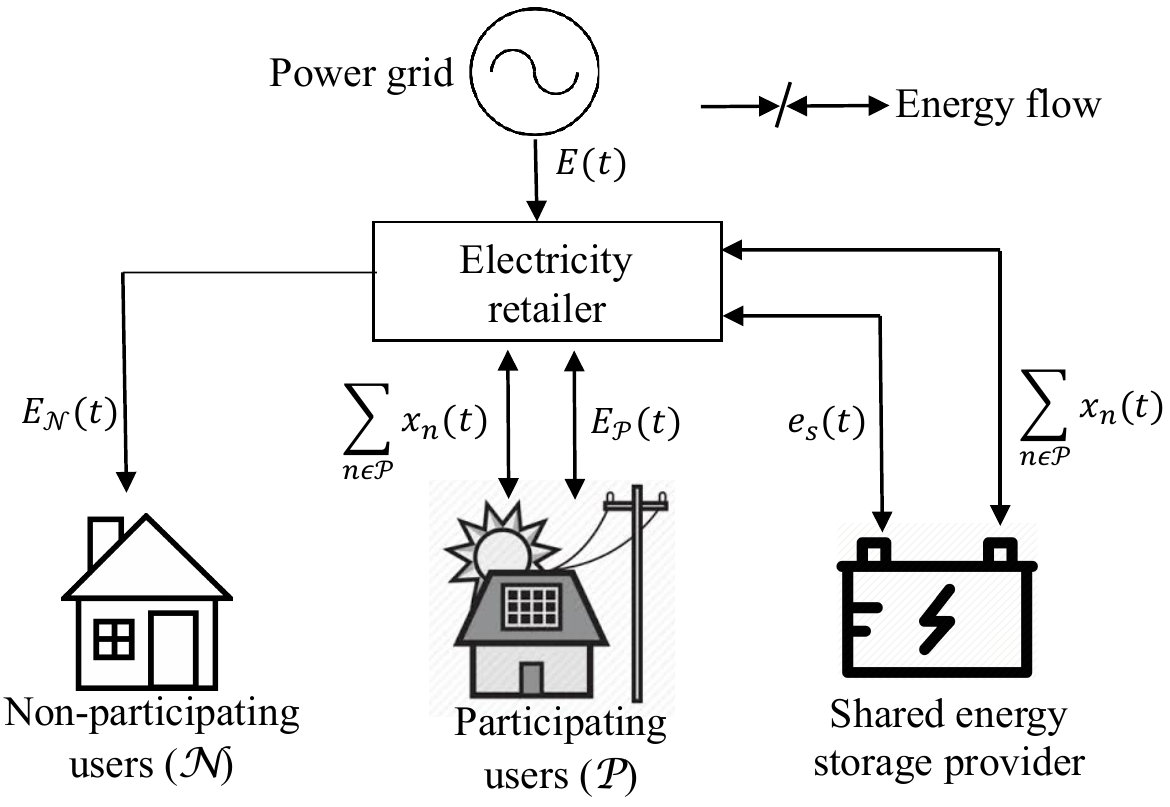}
\caption{Configuration of the energy trading system.}
\label{fig:systemModel}
\end{figure}
 \subsection{Demand-side Model}\label{sec:3-1}
The set of all users in the NAN, $\mathcal{A}$, is divided into two sets, participating users $\mathcal{P}$ and non-participating users $\mathcal{N}$. Thus, $\mathcal{A}=\mathcal{P}~\cup~\mathcal{N}$. The users $\mathcal{P}$ have non-dispatchable energy generation systems, for example, rooftop solar panels, without local energy storage facilities. The users $\mathcal{P}$ participate in the system by trading  energy with the grid and the SES through the electricity retailer. The retailer coordinates the energy transactions between the grid and both the users $\mathcal{A}$ and the SES in addition to the energy transactions between the SES and the users $\mathcal{P}$. The SES is shared among the users $\mathcal{P}$ to store surplus energy from their local energy generation. The users $\mathcal{N}$ may have non-dispatchable energy generation systems without storage and do not participate in the energy trading optimization framework. They are considered as the traditional energy users of the grid. If a user in $\mathcal{N}$ owns a local energy generation system, they sell the excess energy generated directly to the grid through the retailer. 

The time period $\mathcal{T}$, typically one day, is divided into $H$ equal time steps and the control time $t = 1,~2,\dotsm,H$. 
The users $\mathcal{P}$ are divided into two time-dependent sets; $\mathcal{P}^{\scriptstyle{+}}(t)$ and $\mathcal{P}^{\scriptstyle{-}}(t)$. Surplus energy at user $n\in \mathcal{P}$ at time $t$ is given by
\begin{align}
s_n(t)= g_n(t) - d_n(t). \label{eq:id1}
\end{align}
If $g_n(t) > d_n(t)$, i.e., when $s_n(t) > 0$, then user $n\in \mathcal{P}^{\scriptstyle{+}}(t)$. If $g_n(t) < d_n(t)$, i.e., when $s_n(t) < 0$, then user $n\in \mathcal{P}^{\scriptstyle{-}}(t)$. 

For each user $n\in \mathcal{P}$, it is assumed that there is a local energy controlling device. All these controlling devices communicate with a central controlling device at the retailer that performs the energy cost optimization on behalf of the users $\mathcal{P}$. In doing so, the retailer determines users' optimal energy amounts that are traded with the SES and the grid.

For each user $n \in \mathcal{P}$, $e_n(t) > 0$ if they buy energy from the grid, and $e_n(t) < 0$ if they sell energy to the grid. Additionally, $x_n(t) >0 $ if the user sells energy to the SES, and $x_n(t) < 0 $ if they buy energy from the SES. Then, the energy balance at user $n$ gives $e_n(t) = x_n(t) + d_n(t) - g_n(t)$. Hence, with \eqref{eq:id1}, 
\begin{align}
e_n(t)= x_n(t) - s_n(t). \label{eq:id2}
\end{align}
The optimal values for $e_n(t)$ and $x_n(t)$ are determined day-ahead by the retailer using the reported information of $s_n(t)$ by the users $\mathcal{P}$. The users $\mathcal{P}$ generate information of $s_n(t)$ using their next day's energy generation and demand forecasts, and we assume the users $\mathcal{P}$ have accurate energy forecasts \footnote{ To make the game-theoretic analysis tractable, we assume accurate energy forecasts in this paper. On that note, as an interesting future work, to handle energy generation and demand forecast errors, the game-theoretic framework introduced in Section~\ref{sec:4} can be combined with stochastic game theory with imperfect information \cite{gametheoryessentials}.}. Since the surplus energy levels of the users $\mathcal{P}$ are private information and unknown to the retailer, the reported surplus energy level by user $n\in \mathcal{P}$ is denoted by $\hat s_n(t)$ to highlight that the user can misreport this information to the retailer. Then we use $\hat x_n(t)$ and $\hat e_n(t)$ to denote the SES and grid energy transactions determined by the retailer for user $n$ at time $t$, respectively. It is considered that $0\leq \hat x_i(t) \leq \hat s_i(t),~\forall i\in \mathcal{P}^{\scriptstyle{+}}(t)$ and $\hat s_j(t)\leq \hat x_j(t) \leq 0 ,~\forall j\in \mathcal{P}^{\scriptstyle{-}}(t)$. Therefore, with \eqref{eq:id2},
\begin{equation}
\begin{split}
-\hat s_i(t)\leq \hat e_i(t) \leq 0, \quad \forall i \in \mathcal{P}^{\scriptstyle{+}}(t),~t \in \mathcal{T}, \\
0 \leq \hat e_j(t) \leq - \hat s_j(t), \quad \forall j \in \mathcal{P}^{\scriptstyle{-}}(t),~t \in \mathcal{T}. \label{eq:id3}
\end{split}
\end{equation}
\subsection{Shared Energy Storage Model}\label{sec:3-2}
Here, the storage model is similar to that in \cite{Atzeni}. In addition to being charged/discharged with the users $\mathcal{P}$, the SES may exchange energy $e_s(t)$ with the grid at time $t$ through the retailer. In this paper, $e_s(t)>0$ if the SES charges from the grid, and $ e_s(t)<0$ if it discharges energy to the grid.

Consider separating $\hat x_n(t)$ and $e_s(t)$ such that $\hat x_n(t) = \hat x_n^+(t) -\hat x_n^-(t) $ and $e_s(t) = e_s^+(t)- e_s^-(t)$. Here, $\hat x_n^+(t),~e_s^+(t)\geq 0$ are the charging energy profiles, and $\hat x_n^-(t),~e_s^-(t)\geq 0$ are the discharging energy profiles at time $t$. Given that, all optimal SES energy strategies satisfy $\hat x_n^+(t)\hat x_n^-(t) = 0$ and $e_s^+(t)e_s^-(t) =  0$ at each time $t$ to prevent simultaneous charging and discharging of the SES \cite{chathurika2}. We introduce $\eta^+$ and $\eta^-$ such that $0<\eta^+\leq1$ and $\eta^-\geq1$ to consider conversion losses of the SES. For example, if $\hat x^+$ energy is transferred to the SES, only $\eta^+\hat x^+$ energy is effectively stored. Conversely, to get $\hat x^-$ energy, the SES has to be discharged by $\eta^-\hat x^-$. If the charge level of the SES at the beginning of time $t$ is $b(t-1)$, then $b(t)$ is given by
\begin{multline}
b(t)=\alpha b(t-1)+ \eta^+\Big(e_s^+(t) + \sum_{n\in \mathcal{P}}\hat x_n^+(t)\Big)\\
- \eta^-\Big(e_s^-(t) + \sum_{n\in \mathcal{P}}\hat x_n^-(t)\Big). 
\label{eq:id4}
\end{multline}

At each time $t$, $b(t)$ is bounded above by $Q_M$ and below by $0$ \cite{Atzeni}. Therefore, 
\begin{equation}
 0\leq ~b(t)~\leq~ Q_M,~\forall t\in\mathcal{T}\label{eq:id5}
\end{equation}
where $b(t)$ is calculated using \eqref{eq:id4}. To ensure the continuous operation of the SES for the next day and to prevent over-charging or over-discharging of the SES, it is specified \cite{Chen}
\begin{equation}
b(H)=b(0) \label{eq:id6}
\end{equation} 
where $b(0)$ is the initial charge level of the SES that is considered to be within the safe operating region of the SES.
\subsection{Energy Cost Models}\label{sec:3-3}
The retailer buys electricity from the grid at cost $D_t$ at time $t$. If the retailer buys $E(t)$ amount of energy from the grid at time $t$, then $D_t$ is calculated by
\begin{equation}
 D_t = \phi_t E(t)^2+\delta_t E(t) \label{eq:id7}
\end{equation}
where $\phi_t $ and $\delta_t$ are determined according to a day-ahead market clearing process \cite{Atzeni}. The cost function in \eqref{eq:id7} can approximate piecewise linear pricing models used in current electricity markets \cite{Lambotharan} and is widely used in the smart grid literature \cite{MohensianGT, Lambotharan}. 
In the system, $E(t) = e_s(t)+ E_{\mathcal{N}}(t)+E_{\mathcal{P}}(t)$ where $E_{\mathcal{P}}(t)=\sum_{n\in \mathcal{P}} \hat e_n(t)$. Note that it is possible to exist either $E_{\mathcal{N}}(t)\geq 0$ or $E_{\mathcal{N}}(t)\leq 0$ at time $t$. $E_{\mathcal{N}}(t)< 0$ occurs if some or all users in $\mathcal{N}$ have local energy generation systems that produce more energy than the total energy demand of the users $\mathcal{N}$ at time $t$. Using \eqref{eq:id7}, the unit energy price paid by the retailer to the grid at time $t$ is given by
\begin{equation}
 p_g(t)=\phi_t E(t)+\delta_t.\label{eq:id7_1}
\end{equation}
It is considered that $E(t) > 0$ for non-negative pricing in \eqref{eq:id7_1} and assumed $E(t) < E_{\text{max}}$. Additionally, energy transmission losses are neglected within the NAN due to the short proximity between its elements \cite{loss}.
The SES provider sets a price $p_s(t)$ for the traded energy amount $\hat x_n(t)$ by each user $n \in \mathcal{P}$ and trades $e_s(t)$ with the retailer at the price $p_g(t)$.
\section{Energy Trading System}\label{sec:4}
This section explains the energy trading interaction between the SES provider and the retailer by using a non-cooperative Stackelberg game. Stackelberg games are used to explore multi-level decision-making processes. Generally, in a Stackelberg game, one player acts as a leader and moves first to select their strategies. The rest of the players are followers and react to the decisions made by the leader to maximize their profits \cite{vonmarket}. In this paper, the SES provider acts as the leader by moving first to make energy trading decisions, $(p_s(t),~e_s(t)),~\forall t\in\mathcal{T}$, and maximizes revenue. By following the SES provider's decisions, the retailer determines the grid energy allocations, $\hat e_n(t),~ \forall n\in \mathcal{P},~\forall t\in\mathcal{T}$. The solutions to the Stackelberg game are derived by using backward induction \cite{gametheoryessentials}. Thus, the actions of the retailer are derived first based on the knowledge of the  SES provider's actions. Then the analysis proceeds backwards to find the SES provider's actions.
\subsection{Objective of the Electricity Retailer}\label{sec:4-1}
To determine $\hat e_n(t),~ \forall n\in \mathcal{P},~\forall t\in\mathcal{T}$, the retailer minimizes the social cost that is the sum of the total cost for the users $\mathcal{P}$ when they interact with the SES and the cost for the retailer for supplying grid energy to the users $\mathcal{P}$. The social cost at time $t$ is given by
\begin{equation}
C_t = \sum_{n\in \mathcal{P}}\Big({-p_s(t)\hat x_n(t) + p_g(t)\hat e_n(t)}\Big). \label{eq:id8_1}
\end{equation}
By using \eqref{eq:id2} and \eqref{eq:id7_1}, \eqref{eq:id8_1} can be written as a function of $E_{\mathcal{P}} (t)$ such that
\begin{multline}
C_t( E_{\mathcal{P}} (t)) = \phi_t  E_{\mathcal{P}} (t)^2+\Big(\phi_t(E_{\mathcal{N}}(t)+e_s(t))+\delta_t\\
-p_s(t)\Big) E_{\mathcal{P}}(t)
-p_s(t)\sum_{n\in \mathcal{P}}\hat s_n(t). \label{eq:id8_2}
\end{multline}
In response to suitable $(p_s(t),~e_s(t)),~\forall t\in \mathcal{T}$ of the SES provider, the retailer minimizes \eqref{eq:id8_2} at each time $t$ as
\begin{equation}
\min_{E_{\mathcal{P}}(t)~\in~\mathcal{E}}~C_t(E_{\mathcal{P}} (t))  \label{eq:id8}
\end{equation}
where $\mathcal{E}=\prod_{n=1}^I \mathcal{E}_n(t)$, and $\mathcal{E}_n(t)$ subjects to constraints \eqref{eq:id3}. The objective function in \eqref{eq:id8} is strictly convex with respect to $E_{\mathcal{P}} (t)$. Also, its feasible strategy set $\mathcal{E}$ is convex, closed, and non-empty because it is only subject to linear constraints. Hence, \eqref{eq:id8} has a unique solution with respect to $E_{\mathcal{P}}(t)$ \cite{boyd2004convex}.

The unique optimal solution $\tilde  E_{\mathcal{P}}(t)$ in \eqref{eq:id8} for given $(p_s(t), e_s(t))$ is found by using $\frac{\partial C_t(E_{\mathcal{P}} (t))}{\partial E_{\mathcal{P}}(t)} =0$ that gives 
\begin{equation}
\tilde  E_{\mathcal{P}} (t) = \frac{1}{2}\Big(\phi_t^{-1}(p_s(t)-\delta_t)- E_{\mathcal{N}}(t)- e_s(t)\Big). \label{eq:id8a}
\end{equation}

Since $ E_{\mathcal{P}} (t)=\sum_{n\in \mathcal{P}} \hat e_n(t)$, there can be multiple combinations for the grid energy allocation of the users $\mathcal{P}$, $\ve{\hat e}(t) = (\hat e_1(t), \dotsm, \hat e_I(t))$, that satisfy \eqref{eq:id8a}. However, the elements in $\ve{\hat e}(t)$ should satisfy \eqref{eq:id3}. Hence, any tuple of $\ve{\hat e}(t)$ that satisfies both \eqref{eq:id8a} and \eqref{eq:id3} forms the optimal solution for \eqref{eq:id8}. Once $\tilde  E_{\mathcal{P}}(t)$ is found by using \eqref{eq:id8a}, the corresponding grid energy allocation of the users $\mathcal{P}$ at time $t$, $\tilde{\hat {\ve{e}}}(t)$, is found by
\begin{equation}
    \tilde {\hat e}_n(t)=\left\{
                \begin{array}{ll}
                  \frac{ \tilde E_{\mathcal{P}}(t)\hat s_n(t)}{\sum_{n\in \mathcal{P}}\hat s_n(t)},~~\text{if all}~\mathcal{P}~\text{are either surplus or deficit},\\
                  0,~~~~~~~~~~~~~~\text{if $\mathcal{P}$ has both types of users.}  
                \end{array}
              \right. \label{eq:id8b}
\end{equation}
Once $\tilde {\hat e}_n(t)$ is found, the corresponding $\tilde {\hat x}_n(t)$ can be found using \eqref{eq:id2}. Clearly, \eqref{eq:id8a} is a function of $p_s(t)$ and $ e_s(t)$ and hence, it does not guarantee that any values for $(p_s(t), e_s(t))$ would result in $\tilde {\hat e}_n(t)$ satisfying \eqref{eq:id3}. Therefore, constraints are considered in the SES provider's revenue maximization so that its selection of $(p_s(t)$, $ e_s(t))$ ensures $\tilde {\hat e}_n(t)$ in \eqref{eq:id8b} satisfies constraints \eqref{eq:id3}. The details are given in Section~\ref{sec:4-3}.

The users $\mathcal{P}$ may not disclose true surplus energy levels, and this may lead to inefficient results in the above optimization \cite{mechanismBook}. Hence, we are interested in a mechanism that motivates the users $\mathcal{P}$ to report their true information about $s_n(t)$ to the retailer or, in other words, an incentive-compatible mechanism. In game theory context, a mechanism is said to be incentive-compatible if each user can achieve minimized personal cost by being truthful in their actions \cite{mechanismBook}. To this end, we are interested in utilizing the VCG mechanism that can induce reporting true $s_n(t),~\forall t\in \mathcal{T}$ to the retailer as the dominant strategy for user $n\in \mathcal{P}$. A dominant strategy refers to a strategy where a self-interested user minimizes their personal cost regardless of other users' strategies \cite{mechanismBook}. 
\subsection{Vickrey-Clarke-Groves Mechanism for Energy Trading}\label{sec:4-2} 
The VCG mechanism is a well-known solution for obtaining local information from users in resource allocation settings \cite{mechanismBook}. In a VCG mechanism-based energy allocation setting, users are asked to reveal their energy usage information, e.g., surplus energy levels of the users $\mathcal{P}$, to determine energy prices charged to users. In particular, the payments are designed such that users are motivated to reveal local information truthfully.

At time $t$, user $n\in \mathcal{P}$ pays $k_n (t)$ to the retailer for their grid energy transaction $\hat e_n(t)$. Then, for each user $n\in \mathcal{P}$,
\begin{equation}
C_n(t) = -p_s(t)\hat x_n(t) + k_n (t) \label{eq:id9}
\end{equation}
where $-p_s(t)\hat x_n(t)$ is the cost paid to the SES provider. 

We take $\ve{s}^*(t)=(\hat{s}_1(t),\dotsm, s_n(t),\dotsm, \hat{s}_I(t))$ to denote the case when user $n$ reports their true surplus energy level. In response to any tuple $\ve{\hat{s}}(t) =  (\hat{s}_1(t),\ldots,\hat{s}_I(t))$, the retailer, by using the VCG mechanism, selects the grid energy allocation $\ve{\hat e}\big(\ve{\hat{s}}(t) \big)\equiv \ve{\hat e}(t) $ by solving \eqref{eq:id8} as described in Section \ref{sec:4-1} and structures the payments $k_n (\ve{\hat{s}}(t))\equiv k_n (t)$ as
\begin{multline}
k_n(\ve{\hat{s}}(t)) = p_g(\ve{\hat{s}}(t))\hat e_n(\ve{\hat{s}}(t)) + \sum_{m\in \mathcal{P}\backslash n}
\Big({-p_s(t)\hat x_m(\ve{\hat{s}}(t))} \\
+ p_g(\ve{\hat{s}}(t))\hat e_m(\ve{\hat{s}}(t))\Big) - h_n(\ve{\hat{s}}_{-n}(t)).\label{eq:id10}
\end{multline}
Here, $h_n(\ve{\hat{s}}_{-n}(t))$ is a function that depends on the declared surplus energy profile by all other users $ \mathcal{P}\backslash n$ to the retailer at time $t$ given by $\ve{\hat{s}}_{-n}(t) = (\hat{s}_1(t),\ldots,\hat{s}_{n-1}(t),\hat{s}_{n+1}(t),\ldots,\hat{s}_I(t))$. To determine $h_n(\ve{\hat{s}}_{-n}(t))$, we use the popular choice referred to as Clarke tax \cite{mechanismBook}. Then we take
\begin{multline}
h_n(\ve{\hat{s}}_{-n}(t)) = \sum_{m\in \mathcal{P}\backslash n}\Big(-p_s(t)\hat x_m(\ve{\hat{s}}_{-n}(t)) \\
+~p_g(\ve{\hat{s}}_{-n}(t))\hat e_m(\ve{\hat{s}}_{-n}(t))\Big).  \label{eq:id12}
\end{multline}
Here, $\hat e_m(\ve{\hat{s}}_{-n}(t))$ and $\hat x_m(\ve{\hat{s}}_{-n}(t))$ denote the grid and SES energy transactions of user $m\in \mathcal{P}\backslash n$ at time $t$, respectively, after solving \eqref{eq:id8} without user $n$ in the system. Similarly, $p_g(\ve{\hat{s}}_{-n}(t))$ denotes the grid price paid by the retailer, as per \eqref{eq:id7_1}, after solving \eqref{eq:id8} without user $n$ in the system. The case without user $n$ in the system refers to when user $n$ is neither a participating user nor a non-participating user, i.e., when $n \notin \mathcal{A}$. Thus, \eqref{eq:id12} implies $h_n(\ve{\hat{s}}_{-n}(t))$ is the social cost of the other users $\mathcal{P}\backslash n$ in \eqref{eq:id8} without user $n$ in the system. By directly following this argument, we can rewrite \eqref{eq:id10} as
\begin{equation}
k_n(\ve{\hat s}(t)) = p_g(\ve{\hat s}(t))\hat e_n(\ve{\hat s}(t)) + h_n(\ve{\hat{s}}(t)) - h_n(\ve{\hat{s}}_{-n}(t)) \label{eq:id13_1}
\end{equation}
where $h_n(\ve{\hat{s}}(t)) =\sum_{m\in \mathcal{P}\backslash n} \Big({-p_s(t)\hat x_m(\ve{\hat s}(t))} + p_g(\ve{\hat s}(t))\hat e_m(\ve{\hat s}(t))\Big)$ is the social cost of the other users $\mathcal{P}\backslash n$ in \eqref{eq:id8} with user $n$ in the system. Hence, $k_n(\ve{\hat s}(t))$ is the difference between the social costs of the other users $\mathcal{P}\backslash n$ in \eqref{eq:id8} with and without user $n$ in the system plus the retailer's cost for supplying $\hat e_n(\ve{\hat{s}}(t))$. 

It is important that when user $n\in \mathcal{P}^{\scriptstyle{+}}(t)$, they receive an income for selling energy $(\text{i.e.,}~ k_n(t) <0 )$, and when user $n\in \mathcal{P}^{\scriptstyle{-}}(t)$, they incur an expense for buying energy $(\text{i.e.,}~ k_n(t) > 0)$. Proposition~1 demonstrates the proposed mechanism yields this property. 
\begin{proposition}
The payment structure in \eqref{eq:id10} results in an income for user $n$ $(k_n(t) <0)$ when $n\in \mathcal{P}^{\scriptstyle{+}}(t)$ and an expense $(k_n(t) > 0)$ when $n\in \mathcal{P}^{\scriptstyle{-}}(t)$.
\end{proposition}
\begin{proof}
Using \eqref{eq:id2}, the last two terms of \eqref{eq:id13_1} can be derived as
\begin{multline}
h_n(\ve{\hat{s}}(t)) - h_n(\ve{\hat{s}}_{-n}(t)) = p_s(t)\Big(E_{\mathcal{P}\backslash n}(\ve{\hat{s}}_{-n}(t)) - E_{\mathcal{P}\backslash n}(\ve{\hat{s}}(t)) \Big)\\
-\Big(p_g(\ve{\hat{s}}_{-n}(t))E_{\mathcal{P}\backslash n}(\ve{\hat{s}}_{-n}(t))-p_g(\ve{\hat s}(t))E_{\mathcal{P}\backslash n}(\ve{\hat{s}}(t))\Big) \label{eq:id14}
\end{multline}
where $E_{\mathcal{P}\backslash n}(\ve{\hat{s}}(t))=\sum_{m\in \mathcal{P}\backslash n}\hat e_m(\ve{\hat{s}}(t))$ and $E_{\mathcal{P}\backslash n}(\ve{\hat{s}}_{-n}(t))=\sum_{m\in \mathcal{P}\backslash n}\hat e_m(\ve{\hat{s}}_{-n}(t))$ are the total grid load of all other users $\mathcal{P}\backslash n$ at time $t$ with and without user $n$ in the system, respectively.
If we denote the total grid load of all participating users at time $t$ in \eqref{eq:id8} without user $n$ in the system by $ E_{\mathcal{P}}(\ve{\hat{s}}_{-n}(t))$, then $ E_{\mathcal{P}}(\ve{\hat{s}}_{-n}(t))= E_{\mathcal{P}\backslash n}(\ve{\hat{s}}_{-n}(t))$ because, in this case, user $n$ is not in the system, i.e., $n \notin \mathcal{A}$. If we denote the total grid load of all participating users at time $t$ in \eqref{eq:id8} with user $n$ in the system by $E_{\mathcal{P}}(\ve{\hat{s}}(t))$, then $ E_{\mathcal{P}}(\ve{\hat{s}}(t))=E_{\mathcal{P}\backslash n}(\ve{\hat{s}}(t)) + \hat e_n(\ve{\hat s}(t))$ since $n \in \mathcal{P}$.

At optimality in \eqref{eq:id8}, the total grid load of the participating users only depends on $\phi_t,~\delta_t,~E_{\mathcal{N}}(t),~p_s(t)$, and $e_s(t)$ (see \eqref{eq:id8a}) that are considered to be given and hence, fixed at time $t$. Hence, $ E_{\mathcal{P}}(\ve{\hat{s}}(t))= E_{\mathcal{P}}(\ve{\hat{s}}_{-n}(t))$. Then, $E(t)$ remains the same at optimality of \eqref{eq:id8} in both cases with and without user $n$ in the system. Therefore, \eqref{eq:id7_1} reflects $p_g(\ve{\hat{s}}_{-n}(t))=p_g(\ve{\hat{s}}_{n}(t))$. Thus, from \eqref{eq:id13_1} and \eqref{eq:id14},
\begin{equation}
k_n(\ve{\hat s}(t)) = p_s(t)\hat e_n(\ve{\hat s}(t)). \label{eq:id16}
\end{equation}
Since the SES price $p_s(t) > 0$, \eqref{eq:id16} implies that $k_n(t) < 0$ when user $n\in \mathcal{P}^{\scriptstyle{+}}(t)$ with $\hat e_n(t) <0$, and $k_n(t) > 0$ when user $n\in \mathcal{P}^{\scriptstyle{-}}(t)$ with $\hat e_n(t) >0$.
\end{proof}

Clearly, \eqref{eq:id16} implies that the retailer adopts the same price as the SES provider for unit grid energy traded by user $n\in \mathcal{P}$. It is assumed that the users $\mathcal{N}$ pay the same unit energy price as the users $\mathcal{P}$ for their grid energy transactions.

By using a similar explanation to Theorem 10.4.2 in \cite{mechanismBook}, we confirm, by Proposition~2, that the VCG payment structure ensures the disclosure of true surplus energy levels to the retailer is a dominant strategy for each user in $\mathcal{P}$ that minimizes their personal cost in the system.

\begin{proposition}
Revealing true surplus energy levels is a dominant-strategy for each user $n\in \mathcal{P}$ when solving \eqref{eq:id8} with the payment structure \eqref{eq:id10} for given feasible $(p_s(t), e_s(t))$.
\end{proposition}

\begin{proof}
Consider user $n$'s problem is to choose the best strategy for $\hat s_n(t)$ that minimizes their personal cost, $C_n(\ve{\hat{s}}(t)) \equiv C_n(t)$, given in \eqref{eq:id9}. As a shorthand, denote $\ve{\hat{s}}(t)=(\hat s_n(t),~\ve{\hat s}_{-n}(t))$ and $\ve{s}^*(t)=(s_n(t),~\ve{\hat s}_{-n}(t))$. The best strategy for user $n$ is the one that solves
\begin{equation}
\min_{\hat s_n(t)} C_n(\ve{\hat{s}}(t)) \equiv \min_{\hat s_n(t)} \Big(-p_s(t)\hat x_n(\ve{\hat{s}}(t)) + k_n(\ve{\hat{s}}(t))\Big).  \label{eq:id17_1}
\end{equation}
where $k_n(\ve{\hat{s}}(t))$ is given by \eqref{eq:id10}.
Since $h_n(\ve{\hat{s}}_{-n}(t))$ does not depend on $\hat s_n(t)$, it is sufficient for user $n$ to solve
\begin{multline}
\min_{\hat s_n(t)}  \Big(-p_s(t)\hat x_n(\ve{\hat{s}}(t)) + p_g(\ve{\hat s}(t))\hat e_n(\ve{\hat s}(t))\\
+ \sum_{m\in \mathcal{P}\backslash n}
\Big({-p_s(t)\hat x_m(\ve{\hat s}(t))} + p_g(\ve{\hat s}(t))\hat e_m(\ve{\hat s}(t))\Big)\Big) \label{eq:id17}
\end{multline}
by excluding $h_n(\ve{\hat{s}}_{-n}(t))$ from \eqref{eq:id17_1}.

Given $\ve{\hat{s}}(t)$, the retailer picks a grid energy allocation $\ve{\hat e}(t) \in  \mathcal{E}$ for the users $\mathcal{P}$ by minimizing the social cost in \eqref{eq:id8_1}. Hence, by disaggregating the cost of user $n$ and the cost of other users $\mathcal{P}\backslash n$ in \eqref{eq:id8_1}, we can rewrite the retailer's problem as
\begin{multline}
\min_{\ve{\hat e}(t)}  \Big(-p_s(t)\hat x_n(\ve{\hat{s}}(t))  + p_g(\ve{\hat s}(t))\hat e_n(\ve{\hat s}(t))\\
+ \sum_{m\in \mathcal{P}\backslash n}
\Big({-p_s(t)\hat x_m(\ve{\hat s}(t))} + p_g(\ve{\hat s}(t))\hat e_m(\ve{\hat s}(t))\Big)\Big). \label{eq:id17_2}
\end{multline}

Clearly, from \eqref{eq:id17} and \eqref{eq:id17_2}, user $n$ and the retailer try to minimize the same objective function however, with respect to $\hat s_n(t)$ and $\ve{\hat e}(t)$, respectively. Hence, the best way for user $n$ to minimize their personal cost is by influencing the retailer to pick an $\ve{\hat e}(t) \in \mathcal{E}$ that minimizes the cost in \eqref{eq:id17}. This implies user $n$ would like to disclose $\hat s_n(t)$ that leads the retailer to pick an $\ve{\hat e}(t) \in \mathcal{E}$ which solves 
\begin{multline}
\min_{\ve{\hat e}(t)}  \Big(-p_s(t)\hat x_n(\ve{s}^*(t)) + p_g(\ve{s}^*(t))\hat e_n(\ve{s}^*(t))\\
+ \sum_{m\in \mathcal{P}\backslash n}
\Big({-p_s(t)\hat x_m(\ve{s}^*(t))} + p_g(\ve{s}^*(t))\hat e_m(\ve{s}^*(t))\Big)\Big). \label{eq:id17_3}
\end{multline}
Hence, user $n\in \mathcal{P}$ leads the retailer to select the choice that they most prefer by reporting true surplus energy level, i.e., $\hat s_n(t) = s_n(t)$, regardless of the reported surplus energy levels by the other users. Hence, revealing true surplus energy at each time $t$ is a dominant strategy for user $n \in \mathcal{P}$.  
\end{proof}
In the system, the retailer coordinates the energy transactions between the grid and the NAN to maintain the energy balance. Therefore, a mechanism wherein the retailer simultaneously benefits, in addition to the SES provider and the users $\mathcal{P}$, leads to a better operating and economically viable environment. The conditions in Proposition~3 assure the proposed VCG mechanism can achieve this in the system. 
\begin{proposition}
The payment \eqref{eq:id10} produces a non-negative payoff for the retailer across $\mathcal{T}$ given that at time $t \in \mathcal{T}$ 
\begin{equation}
\begin{rcases}
               p_s(t) \leq M(t),~~~\text{if}~E_{\mathcal{A}} (t) < 0, \\
               p_s(t) \geq M(t),~~~\text{otherwise.}              
\end{rcases}\label{eq:id19a}
\end{equation}
where $M(t)= \phi_t( e_s(t)+ E_{\mathcal{N}}(t)) + \delta_t$. 
\end{proposition}
\begin{proof}
At time $t \in \mathcal{T}$, the set $\mathcal{P}$ can have either only surplus users, only deficit users, or both types of users. Hence, we consider all these three different cases in the proof. At time $t$, the retailer buys energy from the grid at $p_g(t)$ in \eqref{eq:id7_1}, and the users $\mathcal{A}$ trade grid energy with the retailer at $p_s(t)$ as shown with Proposition~1. However, the SES provider trades $e_s(t)$ at the same grid price $p_g(t)$ and hence, the retailer receives zero payoffs from the SES provider's grid energy transactions. Therefore, the retailer's net payoff at time $t$ is given by
\begin{equation}
U (t) = \big( p_s(t)-p_g(t)\big) E_{\mathcal{A}} (t). \label{eq:id20}
\end{equation}
Note that, here, we focus on the temporal change of sign, i.e., plus or minus, of $E_{\mathcal{P}}(t)$ due to different types of users in $\mathcal{P}$.

First, assume time $t$ has only surplus users in $\mathcal{P}$. Then, $\hat e_n(t) \leq 0,~\forall n \in \mathcal{P}$ and hence, $ E_{\mathcal{P}}(t) \leq 0$.  Thus, in this case, it is possible to exist either $E_{\mathcal{A}} (t) < 0$ or $E_{\mathcal{A}} (t) > 0$ because $E_{\mathcal{A}}(t) = E_{\mathcal{P}}(t) + E_{\mathcal{N}}(t)$ and $E_{\mathcal{N}}(t)\geq 0$ or $E_{\mathcal{N}}(t)\leq 0$. If $E_{\mathcal{A}} (t) < 0$, then $ p_s(t) \leq p_g(t)$ results in a non-negative payoff in \eqref{eq:id20}. If $E_{\mathcal{A}} (t) > 0$, then $ p_s(t) \geq p_g(t)$ gives a non-negative payoff in \eqref{eq:id20}. Given $p_g(t) = \phi_t( \tilde E_{\mathcal{P}} (t) +E_{\mathcal{N}}(t) + e_s(t)) + \delta_t$ where $ \tilde E_{\mathcal{P}} (t)$ is given by \eqref{eq:id8a}, $p_s(t) \leq  p_g(t)$ gives $p_s(t) \leq M(t)$. Similarly, $p_s(t) \geq p_g(t)$ results in $p_s(t) \geq M(t)$.

Now, assume time $t$ has only deficit users in $\mathcal{P}$. Then, $\hat e_n(t) \geq 0,~\forall n \in \mathcal{P}$ and hence, $  E_{\mathcal{P}}(t) \geq 0$. Hence, with similar arguments to the previous case, it is possible to exist either $E_{\mathcal{A}} (t) < 0$ or $E_{\mathcal{A}} (t) > 0$. If $E_{\mathcal{A}} (t) > 0$, then $ p_s(t) \geq p_g(t)$ (or $p_s(t) \geq M(t)$) gives a non-negative payoff in \eqref{eq:id20}, and if $E_{\mathcal{A}} (t) < 0$, then $ p_s(t) \leq p_g(t)$ (or $p_s(t) \leq M(t)$) gives a non-negative payoff in \eqref{eq:id20}. 

Finally, assume time $t$ has both surplus and deficit users in $\mathcal{P}$. According to \eqref{eq:id8b}, $\hat e_n(t) = 0,~\forall n \in \mathcal{P}$ that implies the payoff in \eqref{eq:id20} is equal to $\big( p_s(t)- p_g(t)\big)E_{\mathcal{N}}(t)$. If $E_{\mathcal{N}} (t) \geq  0$, then $ p_s(t) \geq p_g(t)$ (or $p_s(t) \geq M(t)$) gives a non-negative payoff in \eqref{eq:id20}, and if $E_{\mathcal{N}} (t) \leq 0$, then $ p_s(t) \leq p_g(t)$ (or $p_s(t) \leq M(t)$) gives a non-negative payoff in \eqref{eq:id20}. 

Under these conditions, the sum of $U(t)$ across $\mathcal{T}$ leads to a non-negative payoff for the retailer.
\end{proof}
Therefore, to obtain a non-negative payoff for the retailer, the SES provider's selection of $e_s(t)$ and $p_s(t)$ can be adjusted such that the conditions in \eqref{eq:id19a} are satisfied.
\subsection{Objective of the Shared Energy Storage Provider}\label{sec:4-3}
The SES provider maximizes revenue $R$ by trading energy with the users $\mathcal{P}$ and the grid through the retailer. Then
\begin{equation}
R=\sum_{t=1}^H {\left(-p_s(t)\sum_{n\in \mathcal{P}} \hat x_n(t)-p_g(t)e_s(t)\right)}. \label{eq:id22}
\end{equation}
As per backward induction, by substituting the retailer's strategy \eqref{eq:id8a} in \eqref{eq:id22}, the SES provider's revenue maximization is given by 
\begin{equation}
\max_{\ve{\rho}\in\mathcal{Q}}{\sum_{t=1}^H (\lambda p_s(t)^2+\mu p_s(t)+\nu e_s(t)^2+\xi e_s(t))} \label{eq:id23}
\end{equation}
where $\ve{\rho} = (\ve{p_s},~\ve{ e_s})$ with $\ve{p_s} = (p_s(1),\dotsm, p_s(H))^T$ and $\ve{ e_s} = (e_s(1),\dotsm,e_s(H))^T$. Additionally, $\lambda = -\frac{1}{2}\phi_t^{-1}$, $\mu = (\frac{1}{2}( E_{\mathcal{N}}(t)+\phi_t^{-1}\delta_t)-\sum_{n\in \mathcal{P}}\hat s_n(t))$, $\nu=-\frac{1}{2}\phi_t$, and $\xi =-\frac{1}{2}(\phi_t E_{\mathcal{N}}(t)+\delta_t)$. As recognized in Section \ref{sec:4-1}, to ensure $\tilde {\hat e}_n(t)$ in \eqref{eq:id8b} satisfies \eqref{eq:id3}, the SES provider selects $p_s(t)$ and $e_s(t)$ at time $t$ such that
\begin{equation}
\begin{rcases}
               \sum_{n\in \mathcal{P}} - \hat s_n(t)\leq \tilde{E}_{\mathcal{P}} (t)\leq 0,~~~\text{if all}~\mathcal{P}~\text{are surplus}, \\
                0\leq \tilde{E}_{\mathcal{P}} (t)\leq \sum_{n\in \mathcal{P}} - \hat s_n(t),~~~\text{if all}~\mathcal{P}~\text{are deficit},\\
                \tilde{E}_{\mathcal{P}}(t)=0,~~~~~~~~~~~~~~~~~~~~\text{if}~\mathcal{P}~\text{has both types of users}
\end{rcases}\label{eq:id24}
\end{equation}
where $\tilde{E}_{\mathcal{P}} (t)$ is given by \eqref{eq:id8a}. Thus, the SES provider's feasible strategy set $\mathcal{Q}$ is subject to \eqref{eq:id5}, \eqref{eq:id6}, \eqref{eq:id19a}, and \eqref{eq:id24}. Note that \eqref{eq:id23} is a convex optimization problem that can be solved by using optimization algorithms such as the interior point algorithm \cite{boyd2004convex}. It also has a unique solution because the objective function is strictly convex with respect to $\ve{\rho}$, and $\mathcal{Q}$ is non-empty, closed, and convex due to linear constraints\cite{boyd2004convex}.
\subsection{Non-cooperative Stackelberg Game}
In this system, the Stackelberg game between the SES provider and the retailer can be given by the strategic form of $\mathcal{G}\equiv\langle \{\mathcal{L},\mathcal{F}\},\{\mathcal{Q},\mathcal{E}\},\{R,C_t\}\rangle$
where the SES provider $\mathcal{L}$ is the leader and the retailer $\mathcal{F}$ is the follower.  

\begin{proposition}
The game $\mathcal{G}$ has a unique pure strategy Stackelberg equilibrium.
\end{proposition}

\begin{proof}
In the system, the energy cost minimization \eqref{eq:id8} has a unique solution for given $(p_s(t), e_s(t))$ and that is given by $\tilde  E_{\mathcal{P}} (t)$ (Section~\ref{sec:4-1}). Given $\tilde  E_{\mathcal{P}} (t)$, the optimization problem of the SES provider in \eqref{eq:id23} also has a unique solution (Section~\ref{sec:4-3}). Therefore, according to backward induction, the game $\mathcal{G}$ has a unique pure strategy Stackelberg equilibrium
$(\ve{\rho^*},~\ve{E}^*_{\mathcal{P}})$ where $\ve{E}^*_{\mathcal{P}} = (E^*_{\mathcal{P}}(1),\dotsm,E^*_{\mathcal{P}}(H))$, and $E^*_{\mathcal{P}}(t)$ is found by substituting the $t^{\text{th}}$ element of the optimal solution of \eqref{eq:id23}, i.e., $ \ve{\rho^*}(t) = (p_s^*(t),~ e_s^*(t))$, in \eqref{eq:id8a} \cite{gametheoryessentials}.
\end{proof}

The equilibrium $(\ve{\rho^*},~\ve{E}^*_{\mathcal{P}})$ of the game $\mathcal{G}$ satisfies
\begin{gather}
C_t( E_{\mathcal{P}}^*(t),\ve{\rho^*}) \leq C_t( E_{\mathcal{P}}(t), \ve{\rho^*}),~\forall  E_{\mathcal{P}}(t)\in \mathcal{E},~\forall t\in \mathcal{T} ,\label{eq:id25}
\end{gather}
 \begin{equation}
 \begin{split}
R(\ve{E}^*_{\mathcal{P}},\ve{\rho^*}) \geq R(\ve{E}^*_{\mathcal{P}},\ve{\rho}),~\forall \ve{\rho} \in \mathcal{Q}.
\end{split}\label{eq:id26}
\end{equation}

As shown in Algorithm~\ref{Algo1}, a two-step iterative algorithm, similar to the one in \cite{chathurika2} that was shown to converge, was developed to obtain the Stackelberg equilibrium. 

\begin{algorithm}
\caption{Algorithm to obtain the Stackelberg equilibrium}\label{Algo1}
\begin{algorithmic}[1]
\renewcommand{\algorithmicrequire}{\textbf{Step 1:}}
\renewcommand{\algorithmicensure}{\textbf{Step 2:}}
\REQUIRE
\STATE $r \leftarrow  1$.
\NoThen
\IF{$r=1$}
\STATE{The SES provider selects a feasible starting point for $\ve{\rho}$ and communicates it to the retailer.}
\ELSE
\STATE The SES provider maximizes \eqref{eq:id22} subject to \eqref{eq:id5}, \eqref{eq:id6}, \eqref{eq:id19a}, and \eqref{eq:id24} using $\tilde{E}_{\mathcal{P}} (t),~\sum_{n\in \mathcal{P}} \tilde {\hat x}_n(t),~\forall t\in \mathcal{T}$ and communicates $\ve{\rho}$ to the retailer.
\ENDIF
\ENSURE
\STATE The retailer determines $(\tilde {\hat e}_n(t),~\tilde {\hat x}_n(t)),~\forall n\in \mathcal{P},~\forall t\in \mathcal{T}$ using $\ve{\rho}$, \eqref{eq:id8a}, \eqref{eq:id8b}, and \eqref{eq:id2} and announces $\tilde{E}_{\mathcal{P}} (t),~\sum_{n\in \mathcal{P}} \tilde {\hat x}_n(t),~\forall t\in \mathcal{T}$ back to the SES provider.
\STATE $r \leftarrow r+1$. 
\STATE Repeat from 2 until  ${\|\ve{\rho}^{(r)}-\ve{\rho}^{(r-1)}\|}_2/{\|\ve{\rho}^{(r)}\|}_2 \leq\tau$. Here, $\ve{\rho}^{(r)}$ refers to $\ve{\rho}$ calculated at iteration $r$.
\STATE Return $\tilde{E}_{\mathcal{P}} (t),~\forall t\in \mathcal{T}$ and $\ve{\rho}$ as the equilibrium.
\end{algorithmic}
\end{algorithm}

\section{Results and discussion}\label{sec:5}
To numerically analyze the performance of the proposed system, a residential community of 40 users with photovoltaic (PV) power generation is considered. The demand profiles of the users $\mathcal{P}$ are varied among the users such that their average demand profile is equal to the average daily domestic electricity demand profile in the Western Power Network in Australia on a typical spring day \cite{WPNdata}. For example, when there are 10 participating users, demand profiles of 5 users are generated by scaling down the demand profile in \cite{WPNdata} with scaling factors changing from 0.84 to 0.92 in 0.02 steps. For the other 5 users, the demand profiles are generated by scaling up the demand profile in \cite{WPNdata} with scaling factors changing from 1.08 to 1.16 in 0.02 steps. For each user in $\mathcal{N}$, the same demand profile as the one in \cite{WPNdata} is assigned. For each user in $\mathcal{P}$, PV power generation profile is the same as the average domestic PV power generation profile in the Western Power Network on a typical spring day \cite{WPNdata}. The users $\mathcal{N}$ are considered to be energy users without PV power generation capabilities. We selected $H=48$, $Q_M = 80~\text{kWh}$, $\alpha = 0.9^{(1/48)}$, $\eta^+ = 0.9$, $\eta^- = 1.1$ \cite{Atzeni}, and $b(0) = 0.25Q_M$. In simulations, $\phi$ is selected such that $\phi_{\text{peak}}=1.5\times \phi_{\text{off-peak}}$, and peak period for the energy demand is taken as $16:00-23:00$. Here, $\phi_{\text{peak}}$ is selected such that the range of the predicted price signal, $(p_g(1), \dotsm, p_g(H))$, that the retailer pays to the grid is the same as the range of the time-of-use price signal in \cite{ausgrid}. $\delta_t$ is set as a constant across $\mathcal{T}$ such that the average predicted grid price is the same as the average reference price signal.

\subsection {Preliminary Study of the Proposed System}\label{sec:5-1}
Here, the behavior of the proposed system is compared with a baseline system without an SES where the users $\mathcal{P}$ trade energy only with the grid through the retailer. In the baseline, each user in $\mathcal{P}$ is charged with a price $p_g(t)$,  calculated as per \eqref{eq:id7_1}, for their unit grid energy consumption. 

In Fig.~\ref{fig:25MDPriceVar}, the price signal of the SES provider is equal to that of the retailer used to trade grid energy with the users $\mathcal{P}$ (see \eqref{eq:id16}). At midday, when there is excess PV energy, the SES provider and the retailer in the proposed system pay a higher price to purchase energy from the users $\mathcal{P}$ than in the baseline. During peak demand hours, when there is insufficient PV energy to supply users' energy demand, participation in the proposed system helps to decrease the energy buying price for the users, and the users $\mathcal{P}$ transfer some of their peak energy demand to the SES. In this way, the proposed system reduces the peak energy demand on the grid as illustrated in Fig.~\ref{fig:25MDLaodVar}, and the peak-to-average load ratio is reduced by 30.74\% compared to the baseline. Moreover, the proposed system can deliver nearly 10\% social cost reduction compared to the baseline.

As per Fig.~\ref{fig:25MDPriceVar}, when the users $\mathcal{P}$ have little PV energy (before 09:00 and after 16:00), the retailer sells grid energy to the users at a higher price than the grid price. At other times, the retailer buys energy from the users with a lower price than the grid price. Hence, with $25\%$ participating users, the system provides a cumulative payoff of 520 AU cents to the retailer. 
\begin{figure}[t!]
\centering
\includegraphics[width=0.75\columnwidth]{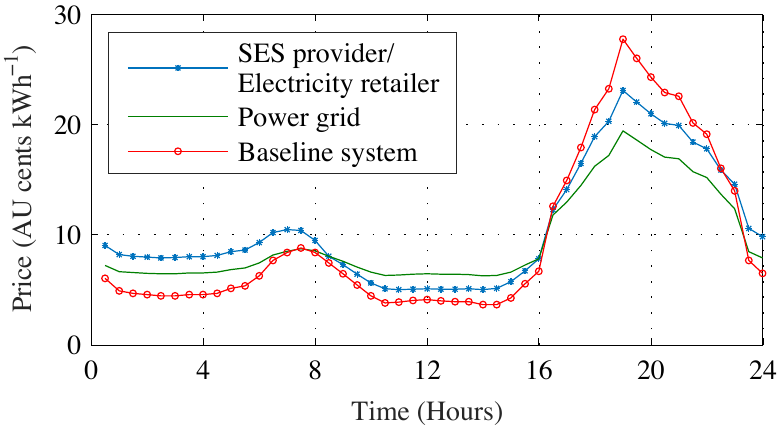}
\caption{Variation of the grid and SES electricity prices in the proposed system compared to the baseline with 25\% participating users (i.e.,10).}
\label{fig:25MDPriceVar}
\end{figure}
\begin{figure}[t!]
\centering
\includegraphics[width=0.75\columnwidth]{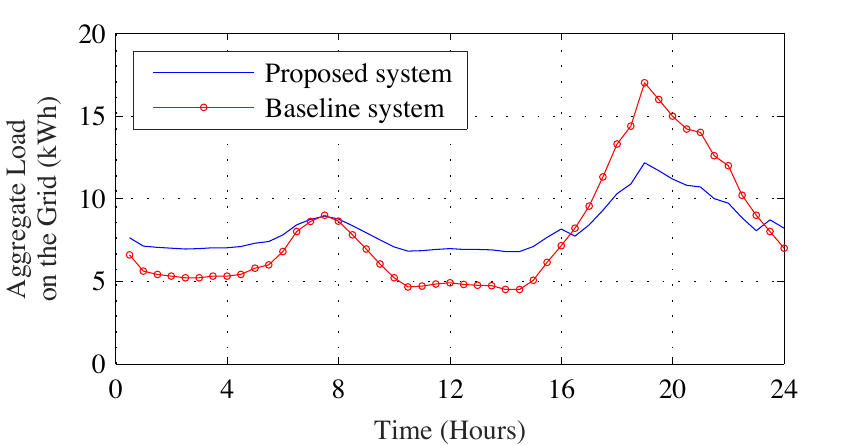}
\caption{Load variation on the grid with 25\% participating users.}
\label{fig:25MDLaodVar}
\end{figure}
\subsection{Participating Users versus Non-participating Users}
The primary incentive for a user to become a participating user is a reduction in their energy cost. Fig.~\ref{fig:NonPU_PU_Costs} depicts the variation of average daily energy costs for the users $\mathcal{P}$ and $\mathcal{N}$ as the number of users in $\mathcal{P}$ increases. The figure shows that the users $\mathcal{P}$ benefit more than the users $\mathcal{N}$ for all percentages of the users $\mathcal{P}$. When the number of users $\mathcal{P}$ increases, the greater demand on the SES leads to a greater price of the SES provider. Concurrently, the grid energy price for the users $\mathcal{A}$ increases because the SES price is equal to the price used by the retailer to trade grid energy with the users $\mathcal{A}$ (see Section~\ref{sec:4-2}). Hence, the average costs for the users $\mathcal{P}$ and $\mathcal{N}$ increase when the number of users $\mathcal{P}$ increases. However, users in $\mathcal{N}$ who pay $127~\text{AU cents}$ on average when there are $10\%$ participating users in the system only have to pay $121~\text{AU cents}$ on average by participating in the system when there are $90\%$ participating users. Hence, it is beneficial for a user in $\mathcal{N}$ to become a user in $\mathcal{P}$ in the proposed system.
\subsection{ Proposed System versus Competitive System }
Here, the performance of the proposed system is compared to the competitive energy trading system proposed in \cite{chathurika2} where the objective is to minimize individual energy costs of the users $\mathcal{P}$ and maximize the revenue for the SES provider. The solution in \cite{chathurika2} is incapable of producing a payoff to an intermediary electricity retailer because there is no price differentiation for grid energy traded by the participating users. In contrast, the proposed system produces a non-negative payoff to a retailer due to the price difference between the grid and the users $\mathcal{A}$ (see \eqref{eq:id20}) while minimizing the social cost in \eqref{eq:id8} and maximizing the revenue for the SES provider.

In Fig.~\ref{fig:CommCost}, the community cost is the sum of the energy costs of the three entities: the SES provider, the users $\mathcal{P}$, and the retailer. Since part of the cost incurred by the SES provider is a profit to the users $\mathcal{P}$ and part of the cost incurred by the users $\mathcal{P}$ produces a payoff to the retailer (see \eqref{eq:id9}, \eqref{eq:id20}, and \eqref{eq:id22}), the community cost in the proposed system is given by $\sum_{t\in \mathcal{T}}p_g(t)\Big( e_s(t)  + E_{\mathcal{A}}(t)\Big) - p_s(t) E_{\mathcal{N}}(t)$. Note that when there are 20\% participating users, energy discharged by the SES to the grid during peak hours leads to a lower total grid energy load, $\Big( e_s(t)  + E_{\mathcal{A}}(t)\Big)$, than that of the users $\mathcal{N}$, $E_{\mathcal{N}}(t)$, during peak hours. As a result, community cost becomes negative with 20\% participating users. In each case in Fig.~\ref{fig:CommCost}, the proposed system has a lower community cost than in the competitive system. The retailer can also benefit in the proposed system whereas the retailer in the competitive system does not receive a positive payoff. Hence, the proposed system provides a more economically viable solution for the energy trading scenario by creating a better trade-off of cost benefits among the SES provider, the users $\mathcal{P}$, and the retailer.
\begin{figure}[t!]
\centering
\includegraphics[width=0.75\columnwidth]{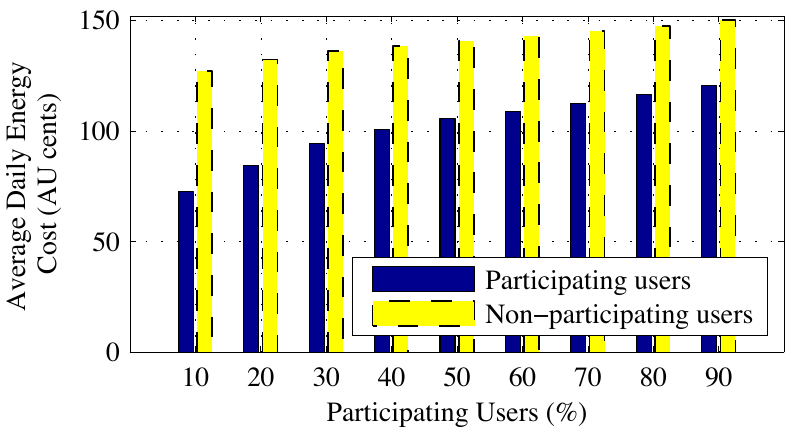}
\caption{Average daily energy costs of participating and non-participating users in the proposed system with different fractions of participating users.}
\label{fig:NonPU_PU_Costs}
\end{figure}
Similar to \cite{chathurika2}, simulation experiments revealed that, in the proposed system, users with more surplus energy in $\mathcal{P}$ pay a lower energy cost compared to users with less surplus energy.

In general, a higher peak-to-average load ratio reduction is preferable as it indicates a flattened load profile on the grid \cite{MohensianGT}. Fig.~\ref{fig:PAR_reduction} depicts when the fraction of participating users in the proposed system increases from $5\%$ to $100\%$, the peak-to-average load ratio reduction on the grid compared to the baseline improves from $28\%$ to $45\%$. The competitive system displays a similar trend because the equilibrium energy trading transactions (the SES and grid energy transactions of the users $\mathcal{P}$) of the two systems present a similar behavior. 
\begin{figure}[t!]
\centering
\includegraphics[width=0.8\columnwidth]{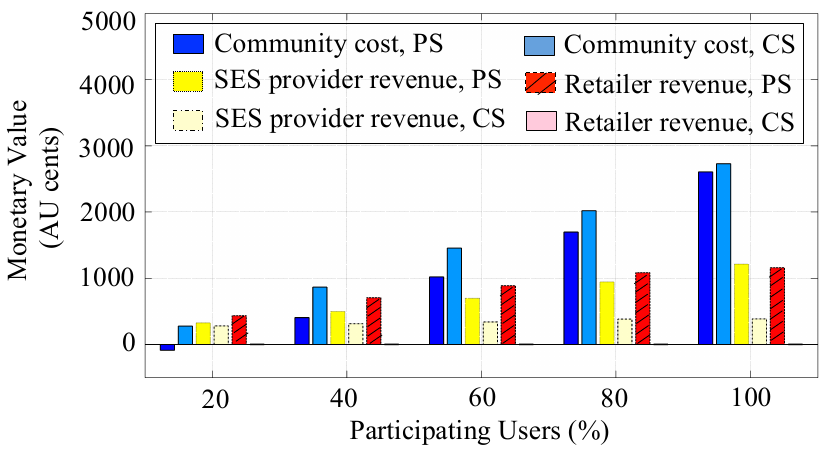}
\caption{Energy cost and revenue comparison with different fractions of participating users. Proposed system (PS), competitive system (CS).}
\label{fig:CommCost}
\end{figure}

\begin{figure}[t!]
\centering
\includegraphics[width=0.8\columnwidth]{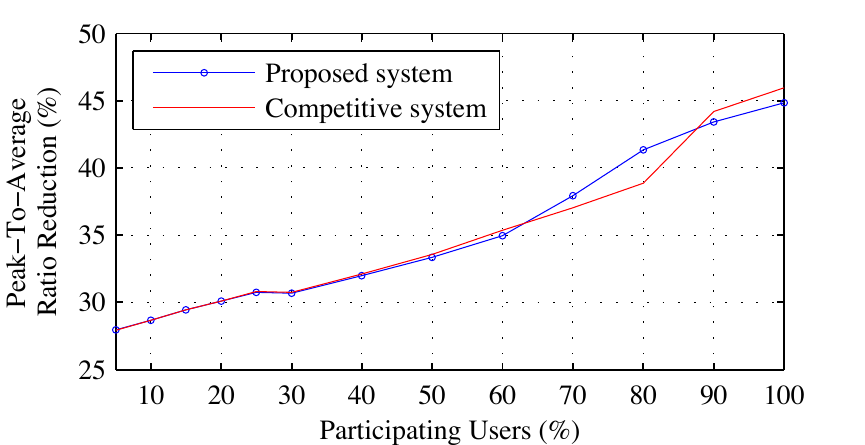}
\caption{Peak-to-average ratio reductions compared to the baseline system with different fractions of participating users.}
\label{fig:PAR_reduction}
\end{figure}
\subsection{ Imperfect Energy Forecasts and the Proposed System}
As a preliminary study in exploring the effects of imperfect energy forecasts on the system performance, proportional variance white noise errors were introduced to PV power generation and energy demand forecasts. When averaged over a large number of simulations, the mean absolute percentage error (MAPE) of energy forecasts is equal to half of percentage white noise variance \cite{chathurika2}. With 60\% participating users in the system, as MAPE increases from 0\% to 50\%, community cost increased from 1020 AU cents to 1319 AU cents, and for each 5\% increase of MAPE, community cost increased by nearly 15 AU cents. This is mainly because the retailer's revenue generated from the non-participating users' energy transactions reduced as MAPE increases. However, as MAPE changes from 0\% to 50\%, the average daily costs for the users $\mathcal{P}$ remained nearly unchanged, at 108 AU cents, with a variance less than 0.0002 AU $\text{cents}^2$. This demonstrates participating user costs are unaffected by the forecast errors. Similar trends were observed when changing the percentage of participating users from 20\% to 100\% in 10\% steps.

\section{Conclusion}\label{sec:6}
The energy trading interaction between a shared energy storage (SES) system and users with non-dispatchable energy generation for demand-side management of a neighborhood area network (NAN) can be studied by minimizing energy costs of users from a social planner's perspective. Here, we have investigated an energy trading system among an SES provider, users with non-dispatchable energy generation, and an electricity retailer for demand-side management of a NAN. By exploiting the insights of the Vickrey-Clarke-Groves mechanism and Stackelberg game theory, we have developed the system to minimize social energy cost for the users and maximize revenue for the SES provider. It has been shown that the proposed system possesses incentive-compatibility for the users by rewarding them in return for disclosing true energy usage information and provides benefits to the retailer in addition to the SES provider and the users.

Future research directions include investigating a stochastic model to incorporate energy forecast errors with the proposed system, extending the system to accommodate additional constraints including secure grid voltage levels with energy generation and demand fluctuations, and exploring a price differentiation strategy between participating and non-participating users to encourage participation in the system.


 \newcommand{\noop}[1]{}



\begin{IEEEbiography}
{Chathurika P. Mediwaththe} (S'12-M'17) received the B.Sc. degree (Hons.) in electrical and electronic engineering from the University of Peradeniya, Sri Lanka. She completed the PhD degree in electrical engineering at the University of New South Wales, Sydney, NSW, Australia in 2017. From 2013-2017, she was a research student with Data61-CSIRO (previously NICTA), Sydney, NSW, Australia. She is currently a research fellow with the Research School of Electrical, Energy and Materials Engineering (RSEEME) and the Battery Storage and Grid Integration Program, Australian National University, Australia. Her current research interests include electricity demand-side management, renewable energy integration into distribution power networks, game theory and optimization for resource allocation in distributed networks, and machine learning. 
\end{IEEEbiography}

\begin{IEEEbiography}
{Marnie Shaw} received the B.Sc. (Hons.) and PhD degrees, both from the School of Physics, University of Melbourne, Australia (2003). Her research interests lie in applying data analytics and machine learning to a range of data-rich problems, including the integration of renewable energy into the electricity grid. She has held postdoctoral research positions at Harvard University and the University of Heidelberg in Germany. She is currently Research Leader in the Battery Storage and Grid Integration Program, at the Australian National University, and convenor of the Energy Efficiency research cluster at the ANUÕs Energy Change Institute.
\end{IEEEbiography}

\begin{IEEEbiography}
{Saman Halgamuge} (F'17) is a Professor in the Department of Mechanical Engineering, School of Electrical, Mechanical and Infrastructure Engineering at the University of Melbourne, an honorary Professor of Australian National University (ANU) and an honorary member of ANU Energy Change Institute.  He was previously the Director of the Research School of Engineering at the Australian National University (2016-18), Professor, Associate Dean International, Associate Professor and Reader and Senior Lecturer at University of Melbourne (1997-2016).  He graduated with Dipl.-Ing and PhD degrees in Data Engineering from Technical University of Darmstadt, Germany and B.Sc. Engineering from University of Moratuwa, Sri Lanka. His research interests are in Machine Learning including Deep Learning, Big Data Analytics and Optimization and their applications in Energy, Mechatronics, Bioinformatics and Neuro-Engineering. His fundamental research contributions are in Big Data Analytics with Unsupervised and Near Unsupervised type learning as well as in Transparent Deep Learning and Bioinspired Optimization. 
\end{IEEEbiography}

\begin{IEEEbiography}
{David B. Smith}  received the B.E. degree in electrical engineering from the University of New South Wales, Sydney, NSW, Australia, in 1997, and the M.E. (research) and Ph.D. degrees in telecommunications engineering from the University of Technology, Sydney, NSW, Australia, in 2001 and 2004, respectively. Since 2004, he was with National Information and Communications Technology Australia (NICTA, incorporated into Data61 of CSIRO in 2016), and the Australian National University (ANU), Canberra, ACT, Australia, where he is currently a Principal Research Scientist with CSIRO Data61, and an Adjunct Fellow with ANU. He has a variety of industry experience in electrical and telecommunications engineering. His current research interests include wireless body area networks, game theory for distributed signal processing, disaster tolerant networks, 5G networks, IoT, distributed optimization for smart grid and privacy for networks.  He has published over 100 technical refereed papers. He has made various contributions to IEEE standardisation activity. He is an area editor for IET Smart Grid and has served on the technical program committees of several leading international conferences in the fields of communications and networks. Dr. Smith was the recipient of four conference Best Paper Awards.
\end{IEEEbiography}

\begin{IEEEbiography}
{Paul Scott} received BEng and BSc degree in 2010 at the ANU, and a PhD in Computer Science in 2016 at the ANU. He is now a research fellow in the Research School of Computer Science at the ANU, researching how optimisation and intelligent systems can enable the integration of distributed energy resources and renewable energy into our electricity networks.
\end{IEEEbiography}

\end{document}